\documentclass[12pt]{article}

\usepackage{fullpage} \usepackage[usenames,dvipsnames]{xcolor}
\usepackage[colorlinks=true,linkcolor=NavyBlue,citecolor=teal]{hyperref}
\usepackage{amsmath,amsthm,amsfonts,amssymb} \usepackage{units,braket}
\usepackage{boxedminipage} \usepackage[capitalize]{cleveref}
\usepackage[boxed]{algorithm} \floatname{algorithm}{Figure}

\numberwithin{equation}{section} 

\newtheorem*{theorem*}{Theorem} \newtheorem{theorem}{Theorem}[section]
 \newtheorem{claim}[theorem]{Claim}
\newtheorem*{proposition*}{Proposition}

\newtheorem{proposition}[theorem]{Proposition}

\DeclareMathOperator{\supp}{supp} \newcommand{\eps}{\epsilon}

\newcommand{\half}{{\textstyle \frac12}}

\newcommand{\Ent}[1]{H\left[{#1}\right]}
\newcommand{\EntSub}[2]{H_{#1}\left[{#2}\right]}
\renewcommand{\P}[1]{{\mathbb{P}}\left[{#1}\right]}
\newcommand{\Psub}[2]{{\mathbb{P}}_{#1}\left[{#2}\right]}
 \newcommand{\R}{\mathbb R}
\newcommand{\C}{\mathbb C} \newcommand{\Z}{\mathbb Z}
 \newcommand{\B}{\{ 0,1 \}}

\begin{document}
\date{} \title{Testing Booleanity and the Uncertainty Principle}
\author{Tom Gur\thanks{Research supported by an Israel Science
    Foundation grant and by the I-CORE Program of the Planning and
    Budgeting Committee and the Israel Science Foundation.} \vspace{-1.5ex}\\
  Weizmann Institute of Science \vspace{-1.5ex}\\
  {\tt tom.gur@weizmann.ac.il} \and Omer Tamuz\thanks{Supported by ISF
    grant 1300/08. Omer Tamuz is a recipient of the Google Europe
    Fellowship in Social Computing, and this research is supported in
    part by this Google
    Fellowship.} \vspace{-1.5ex}\\
  Weizmann Institute of Science \vspace{-1.5ex}\\
  {\tt omer.tamuz@weizmann.ac.il}}

\maketitle
\begin{abstract}
  Let $f:\{-1,1\}^n \to \R$ be a real function on the hypercube, given
  by its discrete Fourier expansion, or, equivalently, represented as
  a multilinear polynomial. We say that it is Boolean if its image is
  in $\{-1,1\}$.

  We show that every function on the hypercube with a sparse Fourier
  expansion must either be Boolean or far from Boolean. In particular,
  we show that a multilinear polynomial with at most $k$ terms must
  either be Boolean, or output values different than $-1$ or $1$ for a
  fraction of at least $2/(k+2)^2$ of its domain.

  It follows that given oracle access to $f$, together with the
  guarantee that its representation as a multilinear polynomial has at
  most $k$ terms, one can test Booleanity using $O(k^2)$ queries. We
  show an $\Omega(k)$ queries lower bound for this problem.
  
  Our proof crucially uses Hirschman's entropic version of
  Heisenberg's uncertainty principle.
\end{abstract}

\newpage
 \parskip=0.2cm \newpage
 
\section{Introduction}
Let $f$ be a function from $\{-1,1\}^n$ to $\R$. Equivalently, one can
consider functions on $\{0,1\}^n$ or $\Z_2^n$, as we do below. A
natural way to represent such a function is as a multilinear
polynomial. For example:
\begin{align*}
  f(x_1,x_2,x_3) = x_1-2x_2x_3+3.5x_1x_2.
\end{align*}
This representation is called the {\em Fourier expansion} of $f$ and
is extremely useful in many applications
(cf.,~\cite{odonnell12analysis}). The coefficients of the Fourier
expansion of $f$ are called the {\it Fourier transform} of $f$. We
denote the Fourier transform by $\hat{f}$, and think of it too as a
function from $\{-1,1\}^n$ to $\R$.

We say that $f$ is Boolean if $f(x)=1$ or $f(x)=-1$ for all $x$ in its
domain. An interesting question in the field of discrete Fourier
analysis of Boolean functions is the following: what does the fact
that $f$ is Boolean tell us about its Fourier transform $\hat{f}$? Is
there a simple characterization of functions that are the Fourier
transform of Boolean functions?

We propose the following observation that lies at the basis of our
proofs: $f$ is Boolean if and only if the convolution (over $\Z_2^n$)
of $\hat{f}$ with itself is equal to the delta function. This follows
from the convolution theorem, as we show below in
Proposition~\ref{thm:boolean-fourier}.

Equipped with this characterization, we consider the question of
determining whether or not $f$ is Boolean. In particular, we consider
the case that we are given black box access to a function $f$,
together with the guarantee that its representation as a multilinear
polynomial has at most $k$ terms, in which case we say that $f$ is
{\em $k$-sparse}. Sparse functions on the hypercube have been the
subject of numerous studies (see,
e.g.,~\cite{nisan1994rank,gopalan2011testing,mansour1994learning}).

We show that $O(k^2)$ queries to $f$ suffice to answer this question
correctly with high probability. This follows from the following
combinatorial result: in Theorem~\ref{thm:bool-is-or-far-from} we show
that if $f$ is not Boolean then it is not Boolean for at least a
$2/(k+2)^2$ fraction of its domain. More generally, we show that for
any set $D \subset \R$ of size $d$, either the image of $f$ is
contained in $D$, or else $f(x) \not \in D$ for at least a
$d!/(k+d)^d$ fraction of the domain of $f$.  We prove an $\Omega(k)$
lower bound for this problem.

Booleanity testing bears resemblance to problems of property testing
of functions on the hypercube (see,
e.g.,~\cite{blais2009,fischer2002,fischer_eldar2002,newman2000}). See
Section~\ref{sec:property testing} below for further discussion.

Our proofs rely on the discrete version of {\em Heisenberg's
  uncertainty principle}. There have been very few applications of the
discrete uncertainty principle in Computer Science, and in fact we are
only familiar with one other such result, concerning circuit lower
bounds~\cite{Jansen06anon-linear}. We expect that more applications
can be found, in particular in cryptography.  See
Sections~\ref{sec:uncertainty} and~\ref{sec:discussion} below for
further discussion.

In the following Section~\ref{sec:main-results} we present our main
results, and in Sections~\ref{sec:boolean-fourier},
\ref{sec:uncertainty}, \ref{sec:property testing}
and~\ref{sec:discussion} we elaborate on the background and relation
to other work, as well as propose a relaxation of our main
claim. Section~\ref{sec:definitions} contains formal definitions, and
proofs appear in Section~\ref{sec:results}.

\subsection{Main results}
\label{sec:main-results}
A function $f :\{-1,1\}^n\to \R$ is $k$-sparse if it can be
represented as a multilinear polynomial with at most $k$ terms. Recall
that we say that $f$ is Boolean if its image is contained in
$\{-1,1\}$.

The following theorem is a combinatorial result, stating that a
function with a sparse Fourier expansion is either Boolean or far from
Boolean.
\begin{theorem}
  \label{thm:bool-is-or-far-from}
  Every $k$-sparse function $f$ is either Boolean, or satisfies
  \begin{equation*}
    \Psub{x}{f(x) \not \in \{-1,1\}} \geq \frac{2}{(k+2)^2}
  \end{equation*}
  where $\Psub{x}{\cdot}$ denotes the uniform distribution over the
  domain of $f$.
\end{theorem}

We in fact prove a more general result:
\begin{theorem}
  \label{thm:is-or-far-from}
  Let $D \subset \R$ be a set with $d$ elements. Then, for any
  $k$-sparse function $f$, one of the following holds.
  \begin{itemize}
  \item Either $\Psub{x}{f(x) \in D} = 1$,
  \item or $\Psub{x}{f(x) \not \in D} \geq \frac{d!}{(k+d)^d}$,
  \end{itemize}
  where $\Psub{x}{\cdot}$ denotes the uniform distribution over the
  domain of $f$.
\end{theorem}
That is, either $f$'s image is in $D$, or it is far from being in
$D$. In particular, for $D=\{-1,1\}$ (or $\{0,1\}$, or any other set
of size two), this theorem reduces to
Theorem~\ref{thm:bool-is-or-far-from}

An immediate consequence of Theorem~\ref{thm:bool-is-or-far-from} is
the following result.
\begin{theorem}
  \label{thm:k-algo}
  For every $\eps>0$ there exists a randomized algorithm with query
  (and time) complexity $O(k^2\log(1/\eps))$ that, given $k$ and
  oracle access to a $k$-sparse function $f$,
  \begin{itemize}
  \item returns {\em true} if $f$ is Boolean, and
  \item returns {\em false} with probability at least $1-\eps$ if $f$
    is not Boolean.
  \end{itemize}
\end{theorem}
This result can easily be extended to test whether the image of a
function on the hypercube is contained in any finite set, using
Theorem~\ref{thm:is-or-far-from}.

We prove the following lower bound:
\begin{theorem}
  \label{thm:k-lower-bound}
  Let $A$ be a randomized algorithm that, given $k$ and oracle access
  to a $k$-sparse function $f$,
  \begin{itemize}
  \item returns {\em true} with probability at least $2/3$ if $f$ is
    Boolean, and
  \item returns {\em false} with probability at least $2/3$ if $f$ is
    not Boolean.
  \end{itemize}
  Then $A$ has query complexity $\Omega(k)$.
\end{theorem}

\subsection{The Fourier transform of Boolean functions}
\label{sec:boolean-fourier}

Let $f,g$ be functions from $\Z_2^n$ to $\R$. Their convolution $f*g$
is also a function from $\Z_2^n$ to $\R$ defined by
\begin{align*}
  [f*g](x) = \sum_{y \in \Z_2^n}f(y)g(x+y),
\end{align*}
where the addition ``$x+y$'' is done using the group operation of
$\Z_2^n$. Note that the convolution operator is both associative and
distributive.

An observation that lies at the basis of our proofs is a
characterization of the Fourier transforms of Boolean functions:
$\hat{f} : \Z_2^n \to \R$ is the Fourier transform of a Boolean
function if its convolution with itself is equal to the delta
function; that is,
\begin{align*}
  \hat{f} * \hat{f} = \delta
\end{align*}
(where $\delta: \Z_2^n \to \B$ is given by $\delta(0)=1$, and
$\delta(x)=0$ for every $x \neq 0$).

This is our Proposition~\ref{thm:boolean-fourier}; it follows from the
convolution theorem (see, e.g.,\cite{katznelson2004introduction}).
Equivalently, given a function $f$ on $\Z_2^n$, one can shift it by
acting on it with $x \in \Z_2^n$ by $[xf](y) = f(x+y)$. Hence the
observation above can be stated as follows: If and only if a function
is orthogonal to its shifted self, for all non-zero shifts in
$\Z_2^n$, then it is the Fourier transform of a Boolean function.

\subsection{The uncertainty principle}
\label{sec:uncertainty}
A distribution over a discrete domain $S$ is often represented as a
non-negative function $f: S \to \R^+$ which is normalized in $L_1$,
i.e., $\sum_{x \in S}f(x)=1$.

In Quantum Mechanics the state of a particle on a domain $S$ is
represented by a {\em complex} function on $S$, and the probability to
find the particle in a particular $x \in S$ is equal to
$|f(x)|^2$. Accordingly, $f$ is normalized in $L_2$, so that $\sum_{x
  \in S}|f(x)|^2=1$.

Often, the domain $S$ is taken to be $\R$ (or some power thereof). In
this continuous case one represents the state of a particle by a
function $f:\R \to \C$ such that $\int_{x \in \R}|f(x)|^2dx =1$, and
then $|f(x)|^2$ is the probability density function of the
distribution of the particle's position. The Fourier transform of $f$,
denoted by $\hat{f}$, is then also normalized in $L_2$ (if one chooses
the Fourier transform operator to be unitary), and $|\hat{f}(x)|^2$ is
the probability density function of the {\em distribution of the
  particle's momentum}.

The Heisenberg uncertainty principle states that the variance of a
particle's position times the variance of its momentum is at least one
- under an appropriate choice of units. Besides its physical
significance, this is also a purely mathematical statement relating a
function on $\R$ to its Fourier transform.

Hirschman~\cite{hirschman1957note} conjectured in 1957 a stronger
entropic form, namely
\begin{align*}
  H_e\Big[f\Big]+\EntSub{e}{\hat{f}} \geq 1-\ln 2,
\end{align*}
where $\EntSub{e}{f} = -\int_{x \in \R} |f(x)|^2\ln |f(x)|^2dx$ is the
differential entropy of $f$. This was proved nearly twenty years later
by Beckner~\cite{beckner1975inequalities}.

When the domain $S$ is $\Z_2^n$ (equivalently, $\{-1,1\}^n$) then a
similar inequality holds, but with a different constant. Let $f:\Z_2^n
\to \C$ have Fourier transform $\hat{f}:\Z_2^n \to \C$. Then
\begin{align*}
  H\Bigg[\frac{f}{||f||}\Bigg]+\Ent{\frac{\hat{f}}{||\hat{f}||}} \geq
  n.
\end{align*}
where $\Ent{f} = -\sum_{x \in \Z_2^n}|f(x)|^2\log_2|f(x)|^2$, and
$\|f\| = \sqrt{\sum_{x \in \Z_2^n}f(x)^2}$. (For a further discussion
on the foregoing inequality, see
Section~\ref{sec:uncertainty-principle}.)

\subsection{Relation to property testing}
\label{sec:property testing}
We note that the problem of testing Booleanity is similar in structure
to a property testing problem. Since its introduction in the seminal
paper by Rubinfeld and Sudan \cite{rubinfeld_sudan1996}, property
testing has been studied extensively, both due to its theoretical
importance, and the wide range of applications it has spanned
(cf.~\cite{goldreich2010,goldreich1996}). In particular, property
testing of functions on the hypercube is an active area of
research~\cite{blais2009,fischer2002,fischer_eldar2002,newman2000}.

A typical formulation of property testing is as follows: Given a fixed
property $P$ and an input $f$, a property tester is an algorithm that
distinguishes with high probability between the case that $f$
satisfies $P$, and the case that $f$ is $\eps$-far from satisfying it,
according to some notion of distance.

The algorithm we present for testing Booleanity given oracle access is
similar to a property testing algorithm. However, in our case there is
no proximity parameter: we show that if a function is not Boolean then
it {\it must} be far from Boolean, and can therefore be proved to not
be Boolean by a small number of queries.  This type of property
testing algorithms have appeared in the context of the study of
adaptive versus non-adaptive testers \cite{gonen_ron10}.

\subsection{Discussion and open questions}
\label{sec:discussion}
In this paper we use a discrete entropy uncertainty principle to prove
a combinatorial statement concerning functions on the hypercube.  To
the best of our knowledge, this is the first time this tool has been
used in the context of theoretical computer science, outside of
circuit lower bounds.

We note that Theorem~\ref{thm:bool-is-or-far-from} and
Theorem~\ref{thm:is-or-far-from} are, in a sense, a dual to the {\em
  Schwartz-Zippel} lemma~\cite{zippel1989explicit,schwartz1980fast}:
both limit the number of roots of a polynomial, given that it is
sparse. Given the usefulness of the Schwartz-Zippel lemma, we suspect
that more combinatorial applications can be found for the discrete
uncertainty principle.

For example, Biham, Carmeli and Shamir~\cite{biham2008bug} show that
an RSA decipherer who uses hardware that has been maliciously altered
can be vulnerable to an attack resulting in the revelation of the
private key. The assumption is that the dechiperer is not able to
discover that it is using faulty hardware, because the altered
function returns a faulty output for only a very small number of
inputs. The uncertainty principle shows that such malicious alteration
is impossible to accomplish with succinctly represented functions:
when the Fourier transform of a function is sparse then it is
impossible to ``hide'' elements in its image.

As for the scope of this study, many questions still remains open. In
particular, there is a gap between the lower bound and the upper bound
for testing Booleanity with oracle access; we are disinclined to guess
which of the two is not tight.

A natural extension of our results is to functions with a Fourier
transform $\hat{f}$ that is not restricted to having support of size
$k$, but rather having {\em entropy} $\log k$; the latter is a natural
relaxation of the former. Unfortunately, we have not been able to
generalize our results given this constraint. However, another natural
constraint which does yield a generalization is the requirement that
the entropy of $\hat{f}*\hat{f}$, the convolution of the Fourier
transform with itself, is at most $2\log k$. See
Proposition~\ref{thm:conv-support} for why this is indeed natural.

Two additional amendments are needed to be added for
Theorem~\ref{thm:bool-is-or-far-from} for it to be thus
generalized. First, we require that $|f|^2=2^n$. Next, recall that we
call a function $f$ Boolean if $f^2=1$. We likewise say that $f$ is
$\eps$-close to being Boolean if
\begin{align*}
  \sqrt{\frac{1}{2^n}\sum_{x \in \Z_2^n}(f(x)^2-1)^2} \leq \eps.
\end{align*}
This is simply the $L_2$ distance of $f^2$ from the constant function
$1$. In the following theorem we do not test for Booleanity, but for
$\eps$-closeness to Booleanity.
\begin{theorem}
  \label{conjecture:entropy}
  Let $\Ent{\frac{\hat{f}*\hat{f}}{\|\hat{f}*\hat{f}\|}} \leq 2\log
  k$, and let $\|f\|^2=2^n$. Then $f$ is either $\eps$-close to
  Boolean, or satisfies
  \begin{equation*}
    \Psub{x}{f(x) \not \in \{-1,1\}} = \Omega\left(\frac{1}{
        k^{2(\eps^2+1)/\eps^2}}\right)
  \end{equation*}
  where $\Psub{x}{\cdot}$ denotes the uniform distribution over the
  domain of $f$.
\end{theorem}

We prove this Theorem in Section~\ref{sec:conditional-proof}.

\section{Definitions}
\label{sec:definitions}
The following definitions are mostly standard. We deviate from common
practice by considering both a function and its Fourier transform to
be defined on the same domain, namely $\Z_2^n$. Some readers might
find $\{0,1\}^n$ or $\{-1,1\}^n$ a more familiar domain for a
function, and likewise the power set of $[n]$ a more familiar domain
for its Fourier transform.

Denote $\Z_2=\Z/2\Z$. For $x,y \in \Z_2^n$ we denote by $x+y$ the sum
using the $\Z_2^n$ group operation. The equivalent operation in
$\{-1,1\}^n$ is pointwise multiplication (i.e., $xy = (x_1y_1, \ldots,
x_ny_n)$).

Let $f:\Z_2^n \to \R$. We denote its $L_2$-norm by
\begin{align}
  \label{eq:norm}
  \|f\| = \sqrt{\sum_{x \in \Z_2^n}f(x)^2},
\end{align}
denote its support by
\begin{align}
  \label{eq:support}
  \supp f = \{x \in \Z_2^n \::\: f(x) \neq 0\},
\end{align}
and denote its entropy by
\begin{align}
  \label{eq:entropy}
  \Ent{f} = -\sum_{x \in \Z_2^n}f(x)^2\log f(x)^2,
\end{align}
where logarithms are base two and $0\log 0 = 0$, by the usual
convention in this case. We remark that for the simplicity of the
presentation, we define norms and convolutions using summation rather
than expectation.

We call a function $f:\Z_2^n \to \R$ {\em Boolean} if its image is in
$\{-1,1\}$, i.e., if $f(x) \in \{-1,1\}$ for all $x \in \Z_2^n$.

Let $\hat{f}:\Z_2^n \to \R$ denote the {\em discrete Fourier
  transform} (also known as the {\em Walsh-Fourier transform} and {\em
  Hadamard transform}) of $f$, or its representation as a multilinear
polynomial:
\begin{align}
  \label{eq:fourier}
  \hat{f}(x) = \frac{1}{2^n}\sum_{y \in \Z_2^n}f(y)\chi_y(x),
\end{align}
where the characters $\chi_y$ are defined by
\begin{align*}
  \chi_y(x) =
  \begin{cases}
    -1 & \sum_{i:y_i=1}x_i = 1\\
    1 & \mbox{otherwise}
  \end{cases}.
\end{align*}
Note that the sum $\sum_{i:y_i=1}x_i$ is over $\Z_2$ and that
$x_i,y_i$ are (respectively) the $i$'th coordinate of $x$ and $y$.  It
follows that the {\em discrete Fourier expansion} of $f$ is
\begin{align}
  \label{eq:inv-fourier}
  f(x) = \sum_{y \in \Z_2^n}\hat{f}(y)\chi_y(x).
\end{align}
Note that this is a representation of $f$ as a multilinear polynomial.
Hence $f:\Z_2^n \to \R$ is $k$-sparse if $|\supp \hat{f}| \leq k$.

We define $\delta : \Z_2^n \to \R$ by
\begin{align*}
  \delta(x) =
  \begin{cases}
    1& \mbox{when } x = (0,\ldots,0)\\
    0& \mbox{otherwise}
  \end{cases}.
\end{align*}
If we denote by ${\bf 1}(x): \Z_2^n \to \R$ the constant function such
that ${\bf 1}(x)=1$ for all $x \in \Z_2^n$, then it is easy to verify
that
\begin{align}
  \label{eq:delta-transform}
  \hat{\bf 1}=\delta.
\end{align}

Given functions $f,g : \Z_2^n \to \R$, their {\em convolution} $f * g$
is also a function from $\Z_2^n$ to $\R$, defined by
\begin{align}
  \label{eq:convolution}
  [f*g](x) = \sum_{y \in \Z_2^n}f(y)g(x+y).
\end{align}
We denote
\begin{align*}
  f^{(2)} = f * f,
\end{align*}
and more generally $f^{(k)}$ is the convolution of $f$ with itself $k$
times. $f^{(0)}$ is taken to equal $\delta$, since $f * \delta = f$.

\section{Proofs}
\label{sec:results}
\subsection{The Fourier transform of Boolean functions}
  
The convolution theorem (see.,
e.g.,~\cite{katznelson2004introduction}) for $\Z_2^n$ states that, up
to multiplication by a constant, the Fourier transform of the
pointwise multiplication of two functions is equal to the convolution
of their Fourier transforms, and that likewise the Fourier transform
of a convolution is the product of the Fourier transforms (again up to
a constant):
\begin{align}
  \label{eq:convolutions}
  \widehat{f \cdot g} = \hat{f} *
  \hat{g},\quad\quad\mbox{and}\quad\quad \widehat{f * g} = 2^n\hat{f}
  \cdot \hat{g}.
\end{align}
The correctness of the constants can be verified by, for example,
setting $f=g={\bf 1}$. The following proposition follows from
Eqs.~\ref{eq:delta-transform} and~\ref{eq:convolutions}.

\begin{proposition}
  \label{thm:boolean-fourier}
  $f :\Z_2^n \to \R$ is Boolean iff $\hat{f}*\hat{f}=\delta$.
\end{proposition}


\subsection{The discrete uncertainty principle}
\label{sec:uncertainty-principle}

The discrete uncertainty principle for $\Z_2^n$ is the following. It
is a straightforward consequence of Theorem 23 in Dembo, Cover and
Thomas~\cite{dembo1991information}; we provide the proof for
completeness, since it does not seem to have previously appeared in
the literature.
\begin{theorem}
  \label{thm:uncertainty}
  For any non-zero function $f:\Z_2^n \to \R$ (i.e., $\|f\| > 0$) it
  holds that
  \begin{align}
    \label{eq:uncertainty}
    H\Bigg[\frac{f}{||f||}\Bigg]+\Ent{\frac{\hat{f}}{||\hat{f}||}}
    \geq n.
  \end{align}
\end{theorem}
\begin{proof}
  Let $U$ be a unitary $n$ by $n$ matrix such that
  $\max_{ij}|u_{ij}|=M$. Let $x \in \C^n$ be such that $\|x\|>0$. Then
  Theorem 23 in Dembo, Cover and Thomas~\cite{dembo1991information}
  states that
  \begin{align*}
    \Ent{\frac{x}{\|x\|}} + \Ent{\frac{Ux}{\|Ux\|}} \geq 2 \log(1/M),
  \end{align*}
  where for $x \in \C^n$ we define $\Ent{x} = -\sum_{i \in
    [n]}|x_i|^2\log |x_i|^2$.

  Let $F$ be the matrix representing the Fourier transform operator on
  $\Z_2^n$. Note that by our definition in Eq.~\ref{eq:fourier}, the
  transform operator $F$ is not unitary. However, if we multiply it by
  $\sqrt{2^n}$ (i.e., normalize the characters $\chi_y$) then it
  becomes unitary. The normalized matrix elements (which are equal to
  the elements of the normalized characters $\chi_y$), are all equal
  to $\pm 1/\sqrt{2^n}$. Hence $M = 1/\sqrt{2^n}$, and
  \begin{align*}
    \Ent{\frac{f}{\|f\|}}+\Ent{\frac{Ff}{\|Ff\|}} \geq 2 \log(1/M) =
    n.
  \end{align*}
\end{proof}

A distribution supported on a set of size $k$ has entropy at most
$\log k$, as can be shown by calculating its Kullback-Leibler
divergence from the uniform distribution (see,
e.g.,~\cite{cover1991elements}). Hence any distribution with entropy
$\log k$ has support of size at least $k$. This fact, together with
the discrete uncertainty principle, yields a proof of the following
claim (see Matolcsi~and~Szucs~\cite{matolcsi1973intersections} or
O'Donnell~\cite{odonnell12analysis} for an alternative proof of
Eq.~\ref{eq:support-uncertainty}.)

\begin{claim}
  \label{thm:support-uncertainty}
  For any non-zero function $f:\Z_2^n \to \R$ (i.e., $\|f\| > 0$) it
  holds that
  \begin{align}
    \label{eq:support-uncertainty}
    |\supp f| \cdot |\supp \hat{f}| \geq 2^n
  \end{align}
  and
  \begin{align}
    \label{eq:support-ent-uncertainty}
    |\supp f| \cdot 2^{\Ent{\hat{f}/ \|\hat{f}\|}} \geq 2^n.
  \end{align}
\end{claim}
\begin{proof}
  By Theorem~\ref{thm:uncertainty} we have that
  \begin{align*}
    \Ent{\frac{f}{||f||}}+\Ent{\frac{\hat{f}}{||\hat{f}||}} \geq n.
  \end{align*}
  Since $\log |\supp(f)| = \log |\supp(f/||f||)| \geq \Ent{f/||f||}$
  then
  \begin{align*}
    |\supp f |\cdot 2^{\Ent{\hat{f}/ \|\hat{f}\|}} \geq 2^n
  \end{align*}
  and likewise
  \begin{align*}
    |\supp f |\cdot|\supp \hat{f} | \geq 2^n.
  \end{align*}
\end{proof}

We note that for the proof of Theorem~\ref{thm:is-or-far-from} we rely
on Claim~\ref{thm:support-uncertainty}, whereas for the more general
Theorem~\ref{conjecture:entropy}, using
Claim~\ref{thm:support-uncertainty} does not suffices and we must use
(the stronger) Theorem~\ref{thm:uncertainty}.

\subsection{Testing Booleanity given oracle access}
We begin by proving the following standard proposition, which relates
the support of functions $f$ and $g$ with the support of their
convolution.
\begin{proposition}
  \label{thm:conv-support}
  Let $g,f :\Z_2^n \to \R$.  Then
  \begin{align*}
    \supp (f*g) \subseteq \supp f + \supp g.
  \end{align*}
\end{proposition}
Here $\supp f + \supp g$ is the set of elements of $\Z_2^n$ that can
be written as the sum of an element in $\supp f$ and an element in
$\supp g$.
\begin{proof}
  Let $x \in \supp (f * g)$. Then, from the definition of convolution,
  there exist $y$ and $z$ such that $f(y) \neq 0$, $g(z) \neq 0$ and
  $x=y+z$. Hence $x \in \supp f + \supp g$.
\end{proof}

We consider a $k$-sparse function $f$ to which we are given oracle
access. We are asked to determine if it is Boolean, or more generally
if its image is in some small set $D$. We here think of $k$ as being
small - say polynomial in $n$.

We first prove the following combinatorial result:
\begin{theorem*}[\ref{thm:is-or-far-from}]
  Let $D \subset \R$ be a set with $d$ elements. Then for any
  $k$-sparse $f$ one of the following holds.
  \begin{itemize}
  \item Either $\Psub{x}{f(x) \in D} = 1$,
  \item or $\Psub{x}{f(x) \not \in D} \geq \frac{d!}{(k+d)^d}$,
  \end{itemize}
  where $\Psub{x}{\cdot}$ denotes the uniform distribution over the
  domain of $f$.
\end{theorem*}
\begin{proof}
  Let $D=\{y_1,\ldots,y_d\}$.  Denote
  \begin{align*}
    g = \prod_{i=1}^d(f-y_i),
  \end{align*}
  so that $g(x) = 0$ iff $f(x) \in D$. Then
  \begin{align*}
    \hat{g} &=
    \left(\hat{f}-y_1\delta\right)*\cdots*\left(\hat{f}-y_d\delta\right)
    = \hat{f}^{(d)}+a_{d-1}\hat{f}^{(d-1)}+ \cdots
    a_1\hat{f}+a_0\delta,
  \end{align*}
  for some coefficients $a_0,\ldots,a_{d-1}$. Therefore
  \begin{align*}
    \supp \hat{g} \subseteq \bigcup_{i=1}^d\supp \hat{f}^{(i)}
    \cup\{0\}.
  \end{align*}
  We show that $|\supp \hat{g}| \leq (k+d)^d/d!$. Let $A = \supp
  \hat{f} \cup \{0\}$. Then by Proposition~\ref{thm:conv-support}
  $\supp \hat{f}^{(i)}$ is a subset of $iA = A+\cdots+A$, where the
  sum is taken $i$ times; this is the set of elements in $\Z_2^n$ that
  can be written as a sum of $i$ elements of $A$. Hence
  \begin{align*}
    \supp \hat{g} \subseteq A \cup 2A \cup \cdots \cup dA.
  \end{align*}
  Since $0 \in A$, then for all $i \leq d$ we have that $iA \subseteq
  dA$. Hence
  \begin{align*}
    \supp \hat{g} \subseteq dA.
  \end{align*}
  Therefore $\supp \hat{g}$ is a subset of the set of elements that
  can be written as the sum of at most $d$ elements of $A$. This
  number is bounded by the number of ways to choose $d$ elements of
  $A$ with replacement, disregarding order. Hence
  \begin{align}
    \label{eq:supp-g}
    |\supp \hat{g}| \leq \binom{|A|-1+d}{d} \leq \frac{(k+d)^d}{d!},
  \end{align}
  since $|A| \leq |\supp \hat{f}| + 1=k+1$.

  Now, if $f(x) \in D$ for all $x \in \Z_2^n$, then clearly
  $\Psub{x}{f(x) \in D} = 1$. Otherwise, $g(x)$ is different than zero
  for some $x$, and so $\|g\| > 0$. Hence we can apply
  Claim~\ref{thm:support-uncertainty} and
  \begin{align*}
    |\supp g| \cdot |\supp \hat{g}| \geq 2^n.
  \end{align*}
  By Eq.~\ref{eq:supp-g} this implies that
  \begin{align*}
    |\supp g| \geq \frac{2^nd!}{(k+d)^d}.
  \end{align*}
  Since the support of $g$ is precisely the set of $x \in \Z_2^n$ for
  which $f(x) \notin D$ then it follows that
  \begin{align*}
    \Psub{x}{f(x) \not \in D} \geq \frac{d!}{(k+d)^d}.
  \end{align*}
\end{proof}

A consequence is that a function that is not Boolean (i.e., the case
$D=\{-1,1\}$) is not Boolean over a fraction of at least $2/(k+2)^2$
of its domain.  Theorem~\ref{thm:k-algo} is a direct consequence of
this result: assuming oracle access to $f$ (i.e., $O(1)$ time random
sampling), the algorithm samples $f$ at random $\half
(k+2)^2\ln(1/\eps)$ times, and therefore will discover an $x$ such
that $f(x) \not \in \{-1,1\}$ with probability at least $1-\eps$ -
unless $f$ is Boolean.

While we were not able to show a tight lower bound, we show that any
algorithm would require at least $\Omega(k)$ queries to perform this
task (even when two-sided error is allowed).

\begin{theorem*}[\ref{thm:k-lower-bound}]
  Let $A$ be a randomized algorithm that, given $k$ and oracle access
  to a $k$-sparse function $f$,
  \begin{itemize}
  \item returns {\em true} with probability at least $2/3$ if $f$ is
    Boolean, and
  \item returns {\em false} with probability at least $2/3$ if $f$ is
    not Boolean.
  \end{itemize}
  Then $A$ has query complexity $\Omega(k)$.
\end{theorem*}

\begin{proof}[Proof of Theorem~\ref{thm:k-lower-bound}]
  Let $A$ be an algorithm that is given oracle access to a function
  $f:\Z_2^n \to \R$, together with the guarantee that $\supp \hat{f}
  \leq k$. When $f$ is Boolean then $A$ returns ``true''. When $f$ is
  not Boolean then $f$ returns ``false'' with probability at least
  $2/3$. We show that $A$ makes $\Omega(k)$ queries to $f$.

  Denote by $B_k$ the set of Boolean functions that depend only on the
  first $\log k$ coordinates. Denote by $C_k$ the set of functions
  that likewise depend only on the first $\log k$ coordinates, return
  values in $\{-1,1\}$ for some $k-1$ of the $k$ possible values of
  the first $\log k$ coordinates, but otherwise return $2$. Note that
  functions in both $B_k$ and $C_k$ have Fourier transforms of support
  of size at most $k$.


  We prove the lower bound on the query complexity of the randomized
  algorithm by showing two distributions, a distribution of Boolean
  functions and a distribution of non-Boolean functions, which are
  indistinguishable to any algorithm that makes a small number of
  queries to the input.  That is, we present two distributions: one
  for which the algorithm should return ``false'' (denoted by
  $\mathcal{D}_0$) and another for which the algorithm should return
  ``true'' (denoted by $\mathcal{D}_1$). We prove that any randomized
  algorithm which performs at most $o(k)$ queries would not be able to
  distinguish between the two distributions with non-negligible
  probability. This proves the claim.

  Let $\mathcal{D}_1$ be the uniform distribution over $B_k$, and let
  $\mathcal{D}_0$ be the uniform distribution over $C_k$.  Observe
  that an arbitrary query to $f$ in either distribution would output a
  non-Boolean value with probability at most $1/k$, independently of
  previous queries with different values of the first $\log k$
  coordinates. Therefore any algorithm that performs $o(k)$ queries
  would find an input for which $f(x)=2$ with probability $o(1)$, and
  would therefore be unable to distinguish between $\mathcal{D}_0$ and
  $\mathcal{D}_1$ with noticeable probability.
\end{proof}

\subsection{Proof of Theorem~\ref{conjecture:entropy}}
\label{sec:conditional-proof}
Recall the statement of Theorem~\ref{conjecture:entropy}.

\begin{theorem*}[\ref{conjecture:entropy}]
  Let $\Ent{\frac{\hat{f}*\hat{f}}{\|\hat{f}*\hat{f}\|}} \leq 2\log
  k$, and let $\|f\|^2=2^n$. Then $f$ is either $\eps$-close to
  Boolean, or satisfies
  \begin{equation*}
    \Psub{x}{f(x) \not \in \{-1,1\}} = \Omega\left(\frac{1}{
        k^{2(\eps^2+1)/\eps^2}}\right)
  \end{equation*}
  where $\Psub{x}{\cdot}$ denotes the uniform distribution over the
  domain of $f$.
\end{theorem*}

We begin by proving a preliminary proposition.
\begin{proposition}
  \label{thm:H-H}
  Let $X$ be a discrete random variable, and let $x_0$ be a value that
  $X$ takes with positive probability. Then
  \begin{align*}
    H(X|X \neq x_0) \leq \frac{H(X)}{\P{X \neq x_0}}.
  \end{align*}
\end{proposition}
\begin{proof}
  Let $A$ be the indicator of the event $X=x_0$. Then
  \begin{align*}
    H(X) & \geq H(X|A) \\ &= \P{X=x_0}H(X|X=x_0) + \P{X\neq x_0}H(X|X
    \neq x_0)\\ &= \P{X\neq x_0}H(X|X \neq x_0),
  \end{align*}
  since $H(X|X=x_0)=0$.
\end{proof}
\begin{proof}[Proof of Theorem~\ref{conjecture:entropy}]
  Assume that $f$ is $\eps$-far from being Boolean. Observe that
  \begin{align}
    \label{eq:conv-eps}
    \|\hat{f}^{(2)}\|^2 = \frac{1}{2^n}\|f^2\|^2 =
    \frac{1}{2^n}\sum_{x \in \Z_2^n}f(x)^4 = \frac{1}{2^n}\sum_{x \in
      \Z_2^n}(f(x)^2-1)^2 + 1 \ge 1+\eps^2,
  \end{align}
  where the equality before last follows from the fact that
  $\|f\|^2=2^n$.

  Let $X$ be a $\Z_2^n$-valued random variable such that $\P{X=x}=
  \hat{f}^{(2)}(x)^2/\|\hat{f}^{(2)}\|^2$.  Since $f$ is normalized,
  then $\hat{f}^{(2)}(0)=1$. Furthermore,
  \begin{align*}
    \P{X \neq 0} = 1-\P{X = 0} =
    1-\frac{\hat{f}^{(2)}(0)^2}{\|\hat{f}^{(2)}\|^2} \ge
    \frac{\eps^2}{\eps^2+1},
  \end{align*}
  since $\hat{f}^{(2)}(0)=1$, and by Eq.~\ref{eq:conv-eps}.

  Let $g=f^2-1$. Then $\hat{g} = \hat{f}^{(2)}-\delta$, $\hat{g}(0) =
  0$, and $\P{X=x|X \neq 0}=\hat{g}(x)^2/\|\hat{g}\|^2$. Hence by
  Proposition~\ref{thm:H-H} it follows that
  \begin{align*}
    \Ent{\frac{\hat{g}}{\|\hat{g}\|}} \leq
    \Ent{\frac{\hat{f}^{(2)}}{\|\hat{f}^{(2)}\|^2}} \cdot
    \frac{\eps^2+1}{\eps^2} \leq 2 \frac{\eps^2+1}{\eps^2}\log k,
  \end{align*}
  where the second inequality follows from the proposition hypothesis
  that
  \begin{equation*}
    \Ent{\frac{\hat{f}*\hat{f}}{\|\hat{f}*\hat{f}\|}} \leq 2\log k.
  \end{equation*}

  By Claim~\ref{thm:support-uncertainty} it follows that
  \begin{equation*}
    |\supp
    g| \cdot 2^{\Ent{\hat{g}/\|\hat{g}\|}} \geq 2^n.
  \end{equation*}
  Hence $|\supp (f^2-1)| \cdot k^{2(\eps^2+1)/\eps} \geq 2^n$, from
  which the proposition\ follows directly, since
  \begin{equation*}
    \Psub{x}{f(x) \not
      \in \{-1,1\}} = \frac{|\supp (f^2-1)|}{2^n}.
  \end{equation*}
\end{proof}

\section{Acknowledgments}
The authors would like to thank Elchanan Mossel for a helpful initial
discussion of the problem and for suggesting the application to
property testing. We would like to thank Adi Shamir for suggesting the
relevance to cryptography, and we would like to thank Oded Goldreich
for discussions regarding the relevance to property testing.  Last, we
would like to thank the anonymous referees for the helpful comments
that allowed us to improve the presentation of the results.

\bibliographystyle{abbrv} \bibliography{booleanity}
\end{document}